\documentclass[letterpaper, 10 pt, conference]{ieeeconf}  

\IEEEoverridecommandlockouts                              

\usepackage{hyperref}

\usepackage{amsmath} 
\usepackage{amssymb}  
\usepackage{graphicx}
\usepackage{style}
\usepackage{url}
\usepackage[noadjust]{cite}
\usepackage{lipsum}
\usepackage[capitalize]{cleveref}
\usepackage{algorithm}

\newtheorem{assumption}{Assumption}
\newtheorem{theorem}{Theorem}
\newtheorem{lemma}{Lemma}
\newtheorem{remark}{Remark}

\title{\LARGE \bf
Distributed Model Predictive Control for Heterogeneous Platoons with Affine Spacing Policies and Arbitrary Communication Topologies
}

\author{Michael H. Shaham$^{1}$ and Ta\c{s}k\i n Pad{\i}r$^{1,2}$
\thanks{Research was sponsored by the DEVCOM Analysis Center and was accomplished under Cooperative Agreement Number W911NF-22-2-001.}
\thanks{$^{1}$Institute for Experiential Robotics,
        Northeastern University, Boston, MA 02116 USA
        {\tt \{shaham.m, t.padir\}@northeastern.edu}}%
\thanks{$^{2}$Ta\c{s}k\i n Pad{\i}r holds concurrent appointments as a Professor of Electrical and Computer Engineering at Northeastern University and as an Amazon Scholar. This paper describes work performed at Northeastern University and is not associated with Amazon.}
}

\begin{document}

\maketitle
\thispagestyle{empty}
\pagestyle{empty}


\begin{abstract}

This paper presents a distributed model predictive control (DMPC) algorithm for a heterogeneous platoon using arbitrary communication topologies, provided each vehicle can communicate with a preceding vehicle in the platoon. The proposed DMPC algorithm can accommodate any spacing policy that is affine in a vehicle’s velocity, which includes constant distance or constant time headway spacing policies. By analyzing the total cost for the entire platoon, a sufficient condition is derived to ensure platoon asymptotic stability. Simulation experiments with a platoon of 50 vehicles and hardware experiments with a platoon of four 1/10th-scale vehicles validate the algorithm and compare performance under different spacing policies and communication topologies. Code for the experiments and a video demonstration of the hardware experiment can be found at~\url{https://www.github.com/river-lab/dmpc_itsc_2024.git}.

\end{abstract}


\section{INTRODUCTION}

Platooning technology promises to improve safety~\cite{axelsson2017safety}, reduce fuel emissions~\cite{liang2016hdvfuel}, and increase traffic throughput~\cite{smith2020throughput}. Since the California PATH program in the 1990s~\cite{shladover1991path}, platooning has received considerable attention around the globe through various efforts such as the SARTRE platooning program~\cite{robinson2010sartre} in Europe, the Energy ITS project~\cite{tsugawa2011energyits} in Japan, and the global Grand Cooperative Driving Challenge that took place in 2011~\cite{ploeg2012gcdc} and 2016~\cite{englund2016gcdc}. These initiatives have laid the groundwork for significant advancements in platooning technology, highlighting its potential benefits.

The platooning problem can be decomposed into four components~\cite{li2015fourcomponent}: vehicle dynamics, communication topology, control algorithms, and spacing policies. It is practical to consider heterogeneous platoons, meaning each vehicle can have different dynamics. The information that each vehicle has access to defines the communication topology. Two common and practical communication topologies are predecessor following (PF), where each vehicle has access to information only from the vehicle in front of it, and bidirectional (BD), where each vehicle has access to information from the vehicles in front of and behind it.

The two most widely studied platooning controllers are linear feedback and distributed model predictive control (DMPC). Linear feedback has the advantage of relying solely on sensor data (in practice, this data comes from a radar, which provides the preceding vehicle's distance and relative speed---data required for adaptive cruise control), whereas DMPC requires communication between vehicles. The distances between vehicles in a platoon are determined by the spacing policy. Linear feedback controllers with a PF topology and a constant distance headway (CDH) spacing policy (where vehicles maintain a predefined distance to its predecessor) are well known to cause string instability~\cite{seiler2004distprop}, \ie, errors seen by the leader propagate and worsen down the platoon. A constant time headway (CTH) policy (where vehicles remain a set time gap behind the predecessor) alleviates this issue~\cite{naus2010stringstable,ploeg2014lpstring}.

Research on linear feedback controllers for platoons began as early as the 1960s~\cite{levine1966optimalerror}. Since then, there have been many theoretical insights related to different spacing policies and communication topologies~\cite{swaroop1996string,naus2010stringstable,hao2013stabdoubleint,zheng2016stability}. More recently, DMPC~\cite{camponogara2002dmpc} has gained attention, and its stability properties have been studied for both linear~\cite{venkat2005stabdmpc} and nonlinear~\cite{dunbar2007drhcnonlinear} systems. Different formulations of the DMPC problem have also been applied to various multi-vehicle control problems, including trajectory optimization~\cite{kuwata2007dmpcmultitraj} and formation stabilization~\cite{dunbar2006drhcformationstab}.

In the context of platoons, the authors of~\cite{dunbar2012drhcplatoon} developed a DMPC controller with a CDH spacing policy that is both asymptotically stable and string stable for PF topologies when each vehicle knows the desired velocity of the platoon a priori. In~\cite{zheng2017distributed}, the work of~\cite{dunbar2012drhcplatoon} was extended by guaranteeing asymptotic stability for arbitrary unidirectional topologies, meaning each vehicle communicates with at least one preceding vehicle, and without followers needing to know the platoon’s goal velocity a priori. More recently, in~\cite{shaham2024design}, two DMPC algorithms were investigated in simulation and on hardware using 1/10th-scale vehicles for PF communication topologies and CDH spacing policies. In this work, the DMPC methods showed favorable performance compared to linear feedback policies and performed significantly better with respect to spacing error as the platoon size scaled up.

This paper proposes a DMPC algorithm that guarantees the asymptotic stability of a platoon using affine spacing policies (including CDH and CTH) and arbitrary communication topologies, provided each vehicle can communicate with a preceding vehicle. Similar to~\cite{zheng2017distributed}, the algorithm does not require prior knowledge of the lead vehicle’s desired velocity. This work builds on ~\cite{zheng2017distributed} by considering a broader set of permissible communication topologies and enabling the use of CTH spacing policies. The proposed algorithm is validated in simulation with 50 vehicles and on hardware using four 1/10th-scale vehicles.


\section{PLATOON MODELING}

\subsection{Vehicle Dynamics}

Consider a heterogeneous platoon of $N + 1$ vehicles indexed by $i = 0, \ldots, N$ where vehicle $0$ is the (virtual) leader and vehicles $1$ through $N$ are the followers. We assume that $j > i$ implies vehicle $i$ precedes vehicle $j$. Each vehicle's dynamics are given by
\begin{gather*}
    x_i(t+1) = A_i x_i(t) + B_i u_i(t), \quad y_i(t) = C x_i(t), \\
    A_i = \bmat{1 & \dt & 0 \\ 0 & 1 & \dt \\ 0 & 0 & 1 - \frac{\dt}{\tau_i}}, \quad B_i = \bmat{0 \\ 0 \\ -\frac{1}{\tau_i}}, \\
    C = \bmat{1 & 0 & 0 \\ 0 & 1 & 0}
\end{gather*}
where $x_i = (p_i,\, v_i,\, a_i) \in \R^n$ is the vehicle's state given by position, velocity, and acceleration; $u_i \in \R^m$ is the control input (desired acceleration); $y_i \in \R^p$ is the output (position and velocity); $\dt$ is the discrete timestep; and $\tau_i$ is the inertial delay of the vehicle's longitudinal dynamics. Note that these dynamics can be obtained via feedback linearization of a nonlinear vehicle longitudinal dynamics model~\cite{zheng2016stability}. All results in this paper are easily extended to the case $C = I$, where $I$ is the identity matrix.

\subsection{Communication Topology}

The communication topology for the platoon is modeled using a graph $\calG = (\calV, \calE)$ with $\calV = \{0, 1, \ldots , N\}$ and $\calE \subseteq \calV \times \calV$. If vehicle $i$ receives information from vehicle $j$, then $(j, i) \in \calE$. To describe the communication topology, we define a set of matrices. First, the adjacency matrices $M,\, M_{\text{pre}} \in \R^{N \times N}$ describe how the following vehicles communicate with one another. The entries of $M$, denoted $m_{i,j}$ for $i,\, j \in \{1, \ldots , N\}$, are given by
\begin{equation*}
    m_{i,j} = \begin{cases}
        1, & (j, i) \in \calE, \\
        0, & \text{otherwise}.
    \end{cases}
\end{equation*}
Note that $M$ only describes communication information between follower vehicles and does not include information about which vehicles communicate with the leader. The entries of $M_\text{pre}$ are the same as $M$ when $j < i$ and are all zero when $j > i$, \ie, $M_\text{pre}$ indicates which preceding vehicles in the platoon a vehicle communicates with and will always be lower triangular with zeros on the diagonal.

Next, we define the diagonal in-degree matrix $D \in \R^{N \times N}$ as
\begin{equation*}
    D = \diag \left( \littlesumx_{j = 1}^N m_{1,j}, \ldots , \littlesumx_{j = 1}^N m_{N,j} \right),
\end{equation*}
\ie, the $i$th diagonal element of $D$ describes how many other follower vehicles vehicle $i$ receives information from. We define $D_\text{pre}$ in a similar way, except the summation occurs over the elements of $M_\text{pre}$. The Laplacian matrices $L,\, L_\text{pre} \in \R^{N \times N}$ are defined as $L = D - M$ and $L_\text{pre} = D_\text{pre} - M_\text{pre}$. Finally, we define the diagonal pinning matrix $P \in \R^{N \times N}$ as $P = \diag(p_1, \ldots , p_N)$, where $p_i \in \{0, 1\}$ is 1 if vehicle $i$ receives information from the leader and is 0 otherwise.

It will also be useful to describe different sets for each follower vehicle $i$. These sets are the receive set $\calR_i$, the share set $\calS_i$, the leader accessible set $\calP_i$, the information set $\calI_i$, and the predecessor information set $\calI_{i, \text{pre}}$, defined as
\begin{equation}
\label{eq:important_sets}
\begin{split}
    \calR_i & = \set{j}{(j, i) \in \calE, j \in \{1, \ldots, N\}} \\
    \calS_i & = \set{j}{(i, j) \in \calE, j \in \{1, \ldots, N\}} \\
    \calP_i & = \begin{cases}
        \{0\}, & \text{if } (0, i) \in \calE \\
        \emptyset, & \text{otherwise}
    \end{cases} \\
    \calI_i & = \calR_i \cup \calP_i \\
    \calI_{i, \text{pre}} & = \set{j \in \calI_i}{j < i}.
\end{split}
\end{equation}

\subsection{Platooning Control Objective}

The goal of platooning is to have each vehicle drive at the same velocity as the leader while maintaining some desired spacing between adjacent vehicles. Let $d_{i,i-1}: \R \to \R$ be an affine function describing the desired distance between vehicle $i$ and its predecessor $i-1$. For a CDH spacing policy, we have $d_{i, i-1}(v) = \delta_{i,i-1}$, where $\delta_{i, i-1} > 0$ is the desired spacing between vehicle $i$ and its predecessor. For a CTH spacing policy with a minimum safety distance $\delta_\text{safe}$, we have $d_{i,i-1}(v) = \delta_h v + \delta_\text{safe}$, where $\delta_h \geq 0$ is the time headway, $v$ is the vehicle's velocity, and $\delta_\text{safe} > 0 $ is the required spacing when the vehicle's are at rest. Thus, the affine spacing policy given by $\delta_h v + \delta_\text{safe}$ is a direct extension of CDH spacing policies. For the rest of this work, we define the spacing policy for each vehicle as
\begin{equation}
    \label{eq:spacing_policy}
    d_{i, i-1}(v_i) = \delta_h v_i + \delta_\text{safe}, \ i = 1, \ldots, N.
\end{equation}
Using~\cref{eq:spacing_policy}, we see that
\begin{equation}
    \label{eq:spacing_to_leader}
    d_{i, 0}(v_i) = \littlesumx_{j=1}^i \left(\delta_h v_j(t) + \delta_\text{safe} \right).
\end{equation}

Note that the desired spacing to the leader depends also on the velocities of preceding vehicles and is therefore unknown unless vehicle $i$ communicates with every preceding vehicle. It is also possible for each vehicle to decide its own time headway and safe distance, but for notational clarity, we use the same quantities for each vehicle. In practice, however, different spacing policies for each vehicle may be useful. For example, we would likely want $\delta_h = \delta_\text{safe} = 0$ for vehicle 1 since we assume vehicle 0 is a virtual leader. Furthermore, if a smaller vehicle is following a heavy-duty vehicle, the time headways could likely be reduced since smaller vehicles can decelerate much quicker than heavy-duty vehicles and avoid accidents in emergency scenarios.

Since each vehicle's desired speed is the leader's speed, we have $v_{i}^\text{des}(t) = v_0(t)$. Applying this with~\cref{eq:spacing_to_leader}, we define the desired output of vehicle $i$ as
\begin{equation}
    \label{eq:desired_output}
    y_{i}^\text{des}(t) = y_0(t) - \bmat{i(\delta_h v_0(t) + \delta_\text{safe}) \\ 0}
\end{equation}
The goal of the platooning controller is to guarantee asymptotic stability, meaning we want to design a controller such that as $t \to \infty$,
\begin{equation*}
    y_i(t) \to y_{i}^\text{des}(t) 
\end{equation*}
for $i = 1, \ldots, N$.


\section{THE DISTRIBUTED MODEL PREDICTIVE CONTROL ALGORITHM}

In this section, we describe the DMPC algorithm. Let $\calI_i = \{i_1, \ldots, i_{\abs{\calI_i}}\}$ and $y_{-i}(t) \defeq (y_{i_1}(t), \ldots, y_{i_{\abs{\calI_i}}}(t))$, \ie, $y_{-i}(t)$ is the vector of information about vehicle $i$'s neighbors at timestep $t$. Next, we define three different trajectories over the prediction horizon $H$ (prediction timesteps are $k=0, \ldots, H$ for states and outputs, and $k = 0, \ldots, H-1$ for inputs) at timestep $t$:
\begin{itemize}
    \item $x_{i,t}^p(k),\, y_{i,t}^p(k),\, u_{i,t}^p(k)$: predicted state/output/input
    \item $x_{i,t}^\star(k),\, y_{i,t}^\star(k),\, u_{i,t}^\star(k)$: optimal state/output/input
    \item $x_{i,t}^a(k),\, y_{i,t}^a(k),\, u_{i,t}^a(k)$: assumed state/output/input
\end{itemize}
Predicted trajectories are the optimization variables in the DMPC optimization problem, optimal trajectories are the optimal solutions to the DMPC optimization problem, and assumed trajectories are the trajectories shared with vehicles in $\calS_i$ and received from vehicles in $\calI_i$.

In DMPC, each vehicle $i = 1, \ldots, N$ solves the following optimal control problem at each timestep $t$:
\begin{equation}
    \label{eq:dmpc_opt}
    \tag{$\calO_i$}
    \begin{array}{ll}
        \text{min.} & J_i(y_{i,t}^p, u_{i,t}^p, y_{i,t}^a, y_{-i,t}^a) \\
        & = \littlesumx_{k=0}^{H-1} l_i(y_{i,t}^p(k), u_{i,t}^p(k), y_{i,t}^a(k), y_{-i,t}^a(k)) \\
        \text{s.t.} & x_{i,t}^p(0) = x_i(t) \\
        & x_{i,t}^p(k+1) = A_i x_{i,t}^p(k) + B_i u_{i,t}^p(k),\ k \in [H-1] \\
        & y_{i,t}^p(k) = C x_{i,t}^p(k),\ k \in [H] \\
        & u_{i,t}^p(k) \in \calU_i,\ k \in [H-1] \\
        & y_{i,t}^p(H) = \frac{1}{\abs{\calI_{i, \text{pre}}}} \littlesumx\limits_{j \in \calI_{i, \text{pre}}} \left( y_{j,t}^a(H) - \tilde d_{i, j}(v_{j,t}^a(H)) \right) \\
        & a_{i,t}^p(H) = 0
    \end{array}
\end{equation}
where $[H] \defeq \{0, 1, \ldots, H\}$, $\calU_i$ is a convex set describing the control limits, and $J_i$ is the cost we are optimizing over that is the sum of stage costs $l_i$. Here, the optimization variables are those with the superscript $p$, and the optimal solution will be denoted with the superscript $\star$. The variables with the superscript $a$ are problem data that depend on the solution to~\cref{eq:dmpc_opt} at timestep $t - 1$. We define the function $\tilde d_{i, j}: \R \to \R^p$ as
\begin{equation}
    \label{eq:ij_distance}
    \tilde d_{i,j}(v) = \bmat{(i - j)\left(\delta_h v + \delta_\text{safe}\right) \\ 0}.
\end{equation}

The stage cost $l_i$ is given by (dropping the time and prediction timestep indexing for brevity)
\begin{multline}
    \label{eq:dmpc_cost}
    l_i(y_i^p, u_i^p, y_i^a, y_{-i}^a) = q_{i,i} \normm{y_i^p - y_i^a} + \\
    \littlesumx_{j \in \calI_i} q_{i,j} \normm{y_i^p - y_j^a + \tilde d_{i, j}(v_i^p)} + r_i \normm{u_i^p}^2
\end{multline}
where $q_{i,j} > 0$ for all $(i, j) \in \calE$ and $r_i > 0$ for all $i$. The norms used in $\ell_i$ can be any norm. Note that $q_{i, i}$ represents the weight vehicle $i$ places on staying near its assumed trajectory, and $q_{i, j}$ represents the weight vehicle $i$ places on tracking vehicle $j$'s assumed output minus the distance offset. How the $q_{i,j}$ are selected will impact the performance and stability of the platoon. For instance, if we select $q_{i,j} = 0$ for all $j \neq i$, then the vehicles in the platoon will not track their neighbors. Further, if we select $q_{i,i} = 0$, the vehicles will not penalize changing their trajectory based on their assumed trajectory, meaning how vehicle $i$ acts could be very different from how vehicle $j \in \calS_i$ expects vehicle $i$ to act.

The full DMPC algorithm is given in~\cref{alg:dmpc}.
\begin{algorithm}
    \caption{Platoon DMPC}\label{alg:dmpc}
    \textbf{Given:} $q_{i,j} > 0 \ \forall (j, i) \in \calE$; $r_i > 0 \ \forall i$, $\delta_h \geq 0$, $\delta_\text{safe} > 0$ \\
    \textbf{Assume:} Each vehicle $i$ has access to $x_i(t)$ at each timestep \\
    \textbf{Initialization:} At timestep $t = 0$, 
    \begin{align*}
        x_{i,0}^a(0) & = x_i(0), \\
        u_{i,0}^a(k) & = 0,\ k = 0, \ldots, H-1, \\
        x_{i,0}^a(k+1) & = A_i x_{i,0}^a(k) + B_i u_{i,0}^a(k),\ k = 0, \ldots, H-1, \\
        y_{i,0}^a(k) & = C x_{i,0}^a(k),\ k = 0, \ldots, H.
    \end{align*}
    \textbf{DMPC loop:} At each timestep $t = 1, 2, \ldots \,$, each vehicle $i=1, \ldots, N$ performs the following:
    \begin{enumerate}
        \item Solve problem~(\ref{eq:dmpc_opt}) with problem data given by $x_i(t)$, $y_{i,t}^a$, and $y_{-i, t}^a$ to obtain $x_{i,t}^\star$, $y_{i,t}^\star$, $u_{i,t}^\star$. \label{item:step_1}
        \item Compute the assumed control input for the next timestep using
        \begin{equation}
            \label{eq:assumed_input_next}
            u_{i,t+1}^a(k) = \begin{cases}
                u_{i,t}^\star(k+1) & k = 0, \ldots, H-2 \\
                0 & k = H-1.
            \end{cases}
        \end{equation}
        Then compute the assumed states and outputs for next timestep using
        \begin{align}
            x_{i,t+1}^a(k) = \begin{cases}
                x_{i,t}^\star(k+1) & k = 0, \ldots, H-1 \\
                A_i x_{i,t}^\star(H) & k = H
            \end{cases} \label{eq:assumed_states_next} \\
            y_{i,t+1}^a(k) = \begin{cases}
                y_{i,t}^\star(k+1) & k = 0, \ldots, H-1 \\
                C x_{i,t+1}^a(H) & k = H.
            \end{cases} \label{eq:assumed_outputs_next}
        \end{align}
        \item Transmit $y_{i,t+1}^a$ to each vehicle in the share set $\calS_i$. Receive $y_{j,t+1}^a$ from each vehicle in the information set $\calI_i$. 
        \item Implement $u_{i,t}^\star(0)$, increment the timestep, and go to~\cref{item:step_1}.
    \end{enumerate}
\end{algorithm}


\section{STABILITY ANALYSIS}

This section derives sufficient conditions on $q_{i,j}$, $i,\, j \in \{1, \ldots, N\}$, to guarantee platoon asymptotic stability.
\begin{assumption}
\label{ass:spanning_tree}
The graph $\calG = (\calV, \calE)$ contains a spanning tree rooted at the leader (meaning there exists a sequence of directed edges that connect vehicle 0 to vehicle $i$ for all $i$), and for each $i = 1, \ldots, N$ there exists $j < i$ such that $(j, i) \in \calE$.
\end{assumption}

This assumption implies that information from the leader will propagate to each vehicle in the platoon, and each vehicle in the platoon receives information from at least one vehicle preceding it. Furthermore, to make the stability argument, we will assume throughout this section that the leader's velocity is constant.


\subsection{Terminal Constraint Analysis}


The terminal constraints in the optimization problem~(\ref{eq:dmpc_opt}) will enable us to use recursive feasibility to prove stability.

\begin{lemma}[\cite{horn_johnson_1985}]
    \label{lemma:kronecker_product_eigenvals}
    Suppose $A \in \R^{n \times n}$ and $B \in \R^{m \times m}$ have eigenvalues $\lambda_1, \ldots, \lambda_n$ and $\mu_1, \ldots, \mu_m$, respectively. Then the eigenvalues of $A \otimes B$ (the Kronecker product of $A$ and $B$) are given by 
    \begin{equation*}
        \lambda_i \mu_j,\ i = 1, \ldots, n,\ j = 1, \ldots, m.
    \end{equation*}
\end{lemma}

\begin{theorem}
    \label{thm:terminal_state_convergence}
    If $\calG$ satisfies~\cref{ass:spanning_tree}, then the terminal state in the optimization problem~(\ref{eq:dmpc_opt}) converges to the desired state relative to the leader in at most $N$ steps, \ie, for each $i = 1, \ldots, N$,
    \begin{equation*}
        y_{i,t}^p(H) = y_{i,t}^\text{des}(H)
    \end{equation*}
    for $t \geq N$.
\end{theorem}
\begin{proof}
    The terminal constraints in~(\ref{eq:dmpc_opt}) are
    \begin{align*}
        y_{i,t}^p(H) & = \tfrac{1}{\abs{\calI_{i, \text{pre}}}} \littlesumx_{j \in \calI_{i, \text{pre}}} \left( y_{j,t}^a(H) - \tilde d_{i, j}(v_{j,t}^a(H)) \right) \\
        a_{i,t}^p(H) & = 0,
    \end{align*}
    and thus $x_{i,t}^\star(H) = x_{i,t}^p(H)$ for all $i = 1, \ldots, H$. Based on these terminal constraints and~\cref{eq:assumed_input_next,eq:assumed_states_next,eq:assumed_outputs_next}, we have $v_{i,t+1}^a(H) = v_{i,t}^p(H)$ and $p_{i,t+1}^a(H) = p_{i,t}^p(H) + \dt v_{i,t}^p(H)$. Combining these facts with~\cref{eq:ij_distance}, we have
    \begin{align}
        y_{i,t+1}^p(H) & = \tfrac{1}{\abs{\calI_{i, \text{pre}}}} \littlesumx_{j \in \calI_{i, \text{pre}}} \left( y_{j,t+1}^a(H) - \tilde d_{i, j}(v_{j,t+1}^a(H)) \right) \notag \\
        & = \tfrac{1}{\abs{\calI_{i, \text{pre}}}} \littlesumx_{j \in \calI_{i, \text{pre}}} \left( \Delta y_{j,t}^p(H) - (i - j) \tilde \delta_\text{safe} \right) \label{eq:next_timestep_horizon}
    \end{align}
    where
    \begin{equation*}
        \Delta = \bmat{1 & \dt - (i - j) \delta_h \\ 0 & 1}, \ \tilde \delta_\text{safe} = \bmat{\delta_\text{safe} \\ 0}.
    \end{equation*}
    Define the error state for each vehicle $i = 1, \ldots, N$ as
    \begin{equation}
        \label{eq:error_state}
        \tilde y_{i,t}^p(H) = y_{i,t}^p(H) - y_{i,t}^\text{des}(H).
    \end{equation}
    Based on~\cref{eq:desired_output}, after rearranging~\cref{eq:error_state} and substituting it into~\cref{eq:next_timestep_horizon}, we obtain
    \begin{multline*}
        \tilde y_{i,t+1}^p(H) + y_{i,t+1}^\text{des}(H) = \\ \tfrac{1}{\abs{\calI_{i, \text{pre}}}} \littlesumx_{j \in \calI_{i, \text{pre}}} \left(\tilde \Delta y_{j,t}^p(H) + y_{i,t+1}^\text{des}(H)\right),
    \end{multline*}
    where we use the fact that $p_{0,t+1}(H) = p_{0,t}(H) + \dt v_{0,t}(H)$ and $v_{0,t+1}(H) = v_{0,t}(H)$. Thus, we obtain
    \begin{equation}
        \label{eq:terminal_output_single_vehicle}
        \tilde y_{i,t+1}^p(H) = \tfrac{1}{\abs{\calI_{i, \text{pre}}}} \littlesumx_{j \in \calI_{i, \text{pre}}} \Delta \tilde y_{j,t+1}^p(H).
    \end{equation}
    Let $\tilde y_{t}^p(H) \defeq (\tilde y_{1,t}^p(H), \ldots, \tilde y_{N,t}^p(H)) \in \R^{Np}$. Then by concatenating~\cref{eq:terminal_output_single_vehicle} and rewriting using matrices, the platoon's terminal output predicted error is given by
    \begin{equation}
        \label{eq:terminal_output_platoon}
        \tilde y_{t+1}^p(H) = \left( \left( \left(D_\text{pre} + P_\text{pre}\right)^{-1} M_\text{pre}\right) \otimes \Delta\right) \tilde y_{t}^p(H).
    \end{equation}
    The diagonal matrix $(D_\text{pre} + P_\text{pre})$ is invertible based on~\cref{ass:spanning_tree}. The adjacency matrix $M_\text{pre}$ is lower triangular with all zeros on the diagonal, and thus has all zero eigenvalues. Hence, the matrix product $(D_\text{pre} + P_\text{pre})^{-1} M_\text{pre}$ is lower triangular with zeros on the diagonal, and therefore has all zero eigenvalues. Based on~\cref{lemma:kronecker_product_eigenvals}, the eigenvalues of $( \left(D_\text{pre} + P_\text{pre}\right)^{-1} M_\text{pre}) \otimes \Delta$ are all zero, and the matrix is nilpotent (a matrix $A \in \R^{n \times n}$ is nilpotent if there exists an integer $k > 0$ such that $A^k = 0$, and the maximum possible value of $k$ is $n$). Thus, we have
    \begin{equation*}
        \tilde y_t^p(H) = 0 \implies y_{i,t}^p(H) = y_{i,t}^\text{des}(H),\ i = 1, \ldots, N 
    \end{equation*}
    for $t \geq N$.
\end{proof}

\begin{remark}
    The use of $\calI_{i, \text{pre}}$ instead of $\calI_i$ in the terminal constraint is crucial in cases where vehicles receive information from vehicles following them, \eg, the BD case. In that scenario, if we had used $\calI_i$ instead of $\calI_{i, \text{pre}}$, the left side of the Kronecker product would be $(D + P)^{-1} M$. The matrix $M$ would not be lower triangular and would have nonzero eigenvalues, so the matrix $((D + P)^{-1} M) \otimes \Delta$ would no longer be nilpotent. In fact, the eigenvalues of $(D + P)^{-1} M$ would have magnitude less than one, so the dynamical system given by~\cref{eq:terminal_output_platoon} would be stable, and the output errors would converge asymptotically. However, this property is not enough to ensure recursive feasibility, which will be required for the stability proofs.
\end{remark}


\subsection{Analysis of the Local Cost Function}


Let $V_i(t)$ be the optimal value of the optimization problem (\ref{eq:dmpc_opt}) for vehicle $i$ at timestep $t$, \ie,
\begin{equation}
    \label{eq:local_cost_function_opt}
    V_i(t) = J_i(y_{i,t}^\star, u_{i,t}^\star, y_{i,t}^a, y_{-i,t}^a).
\end{equation}

\begin{theorem}
    \label{thm:local_cost}
    If $\calG$ satisfies~\cref{ass:spanning_tree} and $t \geq N$, then $\Delta V_i \defeq V_i(t+1) - V_i(t)$ is bounded above using
    \begin{equation*}
        \Delta V_i \leq -l_i\left(y_{i,t}^\star(0), u_{i,t}^\star(0), y_{i,t}^a(0), y_{-i,t}^a(0)\right) + \littlesumx_{k=1}^{H-1} \epsilon_{i,k}
    \end{equation*}
    where
    \begin{equation*}
        \eps_{i,k} \leq \littlesumx_{j \in \calR_i} q_{i,j} \normm{y_{j,t}^\star(k) - y_{j,t}^a(k)} - q_{i,i} \normm{y_{i,t}^\star(k) - y_{i,t}^a(k)}.
    \end{equation*}
\end{theorem}
\begin{proof}
    Assume~\cref{ass:spanning_tree} holds and $t \geq N$. From~\cref{thm:terminal_state_convergence}, we have $y_{i,t}^\star(H) = y_{i,t}^\text{des}(H)$ for all $i$. Based on~\cref{eq:assumed_input_next,eq:assumed_states_next,eq:assumed_outputs_next}, a feasible control sequence for the next timestep $t+1$ is given by $u_{i,t+1}^a(H)$ with corresponding outputs $y_{i,t+1}^a(H)$. Then we can bound the optimal value $V_i$ at the next timestep as follows:
    \begin{align*}
        & V_i(t+1) \\
        \def\arraystretch{1.5}
        & \begin{array}{l}
            \leq J_i\left(y_{i,t+1}^a, u_{i,t+1}^a, y_{i,t+1}^a, y_{-i,t+1}^a\right) \\
            = \littlesumx\limits_{k=0}^{H-1} l_i\left(y_{i,t+1}^a(k), u_{i,t+1}^a(k), y_{i,t+1}^a(k), y_{-i,t+1}^a(k)\right) \\
            = \littlesumx\limits_{k=0}^{H-2} l_i\left(y_{i,t}^\star(k+1), u_{i,t}^\star(k+1), y_{i,t}^\star(k+1), y_{-i,t}^\star(k+1)\right) \\
            = \littlesumx\limits_{k=1}^{H-1} l_i\left(y_{i,t}^\star(k), u_{i,t+1}^\star(k), y_{i,t+1}^\star(k), y_{-i,t+1}^\star(k)\right)
        \end{array}
    \end{align*}
    where the second equality follows from the construction of the assumed inputs, \cref{eq:assumed_input_next}, and outputs, \cref{eq:assumed_outputs_next}, and the third (last) equality comes from changing the index of summation. Using this result, we have
    \begin{align*}
        \Delta V_i \leq & \littlesumx_{k=1}^{H-1} l_i\left(y_{i,t}^\star(k), u_{i,t+1}^\star(k), y_{i,t+1}^\star(k), y_{-i,t+1}^\star(k)\right) \\
        & - \littlesumx_{k=0}^{H-1} l_i\left(y_{i,t}^\star(k), u_{i,t+1}^\star(k), y_{i,t+1}^a(k), y_{-i,t+1}^a(k)\right) \\
        = & -l_i\left(y_{i,t}^\star(0), u_{i,t}^\star(0), y_{i,t}^a(0), y_{-i,t}^a(0)\right) + \littlesumx_{k=1}^{H-1} \epsilon_{i,k}
    \end{align*}
    where 
    \begin{multline}
        \label{eq:eps_i_k_1}
        \eps_{i,k} = l_i\left(y_{i,t}^\star(k), u_{i,t}^\star(k), y_{i,t}^\star(k), y_{-i,t}^\star(k)\right) \\
        - l_i\left(y_{i,t}^\star(k), u_{i,t}^\star(k), y_{i,t}^a(k), y_{-i,t}^a(k)\right).
    \end{multline}
    After expanding the two $l_i$ terms, the terms related to $u_{i,t}^\star(k)$ cancel out. Furthermore, since the virtual leader's trajectory is assumed constant speed, we have $y_{0,t}^\star(k) = y_{0,t}^a(k)$, so those terms cancel out if vehicle $i$ communicates with the leader. Thus, in~\cref{eq:dmpc_cost}, the summation over vehicles in the information set, $\calI_i$, becomes a summation over the vehicles in the receive set, $\calR_i$. Then, dropping the timestep index $k$ for brevity,~\cref{eq:eps_i_k_1} becomes
    \begin{align*}
        \eps_{i,k} = & \ q_{i,i} \normm{y_{i,t}^\star - y_{i,t}^\star} - q_{i,i} \normm{y_{i,t}^\star - y_{i,t}^a} \\
        & + \littlesumx_{j \in \calR_i} \Bigl(q_{i,j} \normm{y_{i,t}^\star - y_{j,t}^\star + \tilde d_{i,j}(v_{i,t}^\star)} \\
        & \qquad \quad \ \, - q_{i,j} \normm{y_{i,t}^\star - y_{j,t}^a + \tilde d_{i,j}(v_{i,t}^\star)} \Bigr).
    \end{align*}
    Removing the first term which is zero and applying the reverse triangle inequality to the two norms inside the summation over $\calR_i$, we obtain
    \begin{equation*}
        \eps_{i,k} \leq \littlesumx_{j \in \calR_i} q_{i,j} \normm{y_{j,t}^\star - y_{j,t}^a} - q_{i,i} \normm{y_{i,t}^\star - y_{i,t}^a}
    \end{equation*}
    as desired.
\end{proof}
\begin{remark}
    Due to the use of $\calI_{i, \text{pre}}$ in the terminal equality constraint, this proof is similar to the proof provided in~\cite{zheng2017distributed}, but with a change of notation and use of the affine spacing policy. As discussed in~\cite{zheng2017distributed}, this result is not strong enough to determine conditions on the weights $q_{i,j}$ to guarantee stability since the output vectors within the two norms are related to vehicles $i$ and $j$. Nevertheless, by using the sum of the local cost functions, we can alleviate this issue.
\end{remark}


\subsection{Analysis of the Sum of Local Cost Functions}


Define $V(t)$ as the sum of the optimal values of the optimization problem (\ref{eq:dmpc_opt}) given by~\cref{eq:local_cost_function_opt}, over $i = 1, \ldots, N$ at timestep $t$, \ie,
\begin{equation*}
    V(t) = \littlesumx_{i=1}^N V_i(t) = \littlesumx_{i=1}^N J_i(y_{i,t}^\star, u_{i,t}^\star, y_{i,t}^a, y_{-i,t}^a).
\end{equation*}
Before providing the next theorem, it is important to note the following fact based on the definitions of the receive set $\calR_i$ and the share set $\calS_i$. For any function $f(i,j)$ that depends on vehicles $i$ and $j$, we have the following equality:
\begin{equation}
    \label{eq:share_receive_equality}
    \littlesumx_{i=1}^N \littlesumx_{j \in \calR_i} f(i,j) = \littlesumx_{i=1}^N \littlesumx_{j \in \calS_i} f(j,i).
\end{equation}

\begin{theorem}
    \label{thm:sum_local_cost}
    If $\calG$ satisfies~\cref{ass:spanning_tree} and $t \geq N$, then $\Delta V \defeq V(t+1) - V(t)$ is bounded above using
    \begin{equation*}
        \Delta V \leq -\littlesumx_{i=1}^N l_i\left(y_{i,t}^\star(0), u_{i,t}^\star(0), y_{i,t}^a(0), y_{-i,t}^a(0)\right) + \littlesumx\limits_{k=1}^{H-1} \epsilon_k
    \end{equation*}
    where
    \begin{equation*}
        \epsilon_k = \littlesumx_{i=1}^N \left[ \left(\littlesumx_{j \in \calS_i} q_{j,i} - q_{i,i} \right) \normm{y_{i,t}^\star(k) - y_{i,t}^a(k)} \right].
    \end{equation*}
\end{theorem}
\vspace{10pt}
\begin{proof}
    Suppose $\calG$ satisfies~\cref{ass:spanning_tree} and $t \geq N$. Then from~\cref{thm:local_cost}, we have
    \begin{equation*}
        \Delta V \leq -\littlesumx_{i=1}^N l_i\left(y_{i,t}^\star(0), u_{i,t}^\star(0), y_{i,t}^a(0), y_{-i,t}^a(0)\right) + \littlesumx_{k=1}^{H-1} \epsilon_k
    \end{equation*}
    where (again, dropping the timestep index $k$ for brevity)
    \begin{align*}
        \epsilon_k & = \littlesumx_{i=1}^N \eps_{i,k} \\
        & = \littlesumx_{i=1}^N \left(\littlesumx_{j \in \calR_i} q_{i,j} \normm{y_{j,t}^\star - y_{j,t}^a} - q_{i,i} \normm{y_{i,t}^\star - y_{i,t}^a}\right) \\
        & = \littlesumx_{i=1}^N \left(\littlesumx_{j \in \calS_i} q_{j,i} \normm{y_{i,t}^\star - y_{i,t}^a} - q_{i,i} \normm{y_{i,t}^\star - y_{i,t}^a}\right) \\
        & = \littlesumx_{i=1}^N \left[ \left(\littlesumx_{j \in \calS_i} q_{j,i} - q_{i,i} \right) \normm{y_{i,t}^\star(k) - y_{i,t}^a(k)} \right].
    \end{align*}
    where the third equality follows from~\cref{eq:share_receive_equality}.
\end{proof}

\subsection{A Sufficient Condition for Platoon Asymptotic Stability}

To ensure asymptotic stability, the Lyapunov candidate function $V$ must always be positive and decreasing between timesteps whenever it does not equate to zero. The function $V$ is clearly always positive since it is the nonnegative weighted sum of norms. It is also only zero when the optimal trajectories are equal to the assumed trajectories, which occurs when the entire platoon has stabilized to the desired trajectories relative to the leader. Thus, we are left with proving that $\Delta V$ is always decreasing whenever the cost at timestep $t$ is not equal to zero. The following theorem addresses this.

\begin{theorem}
    \label{thm:stability_condition}
    If $\calG$ satisfies~\cref{ass:spanning_tree} and $t \geq N$, then the platoon is asymptotically stable if for each $i = 1, \ldots, N$,
    \begin{equation}
        \label{eq:sufficient_condition}
        q_{i,i} \geq \littlesumx_{j \in \calS_i} q_{j,i}.
    \end{equation}
\end{theorem}
\begin{proof}
    Assume $\calG$ satisfies~\cref{ass:spanning_tree}, $t \geq N$, and~\cref{eq:sufficient_condition} holds. Then we have
    \begin{align*}
        \epsilon_k & = \littlesumx_{i=1}^N \left[ \left(\littlesumx_{j \in \calS_i} q_{j,i} - q_{i,i} \right) \normm{y_{i,t}^\star(k) - y_{i,t}^a(k)} \right] \\
        & \leq 0.
    \end{align*}
    Combining this with~\cref{thm:sum_local_cost}, we have
    \begin{equation*}
        V(t+1) - V(t) \begin{cases}
            = 0 & y_{i,t}^\star(k) = y_{i,t}^a(k),\ \forall i, k \\
            < 0 & \text{otherwise}
        \end{cases}
    \end{equation*}
    and all Lyapunov stability conditions are satisfied. As a result, platoon asymptotic stability is guaranteed.
\end{proof}

\begin{remark}
    This condition is different from that originally derived in~\cite{zheng2017distributed} (as well as in future works, like~\cite{qiang2023dmpc}) since we use scalar weights $q_{i,j}$ instead of positive definite matrices $Q_{i,j}$ within the norm. We also allow for the use of arbitrary norms in the cost function. In theorem 5 in~\cite{zheng2017distributed}, the authors use the condition (using the same notation as this paper)
    \begin{equation}
        \label{eq:previous_condition}
        Q_{i,i} \succeq \littlesumx_{j \in \calS_i} Q_{j,i}
    \end{equation}
    to claim that
    \begin{equation}
        \label{eq:previous_claim}
        \littlesumx_{j \in \calS_i} \normm{y_{i,t}^\star(k) - y_{i,t}^a(k)}_{Q_{j,i}} - \normm{y_{i,t}^\star(k) - y_{i,t}^a(k)}_{Q_{i,i}} \leq 0
    \end{equation}
    where $\normm{x}_Q \defeq \sqrt{x^\tp Q x}$. However, this does not always hold if $\abs{\calS_i} > 1$, which we prove by counterexample. Let $z \defeq y_{i,t}^\star(k) - y_{i,t}^a(k)$, and suppose $\calS_i = \{i_1, i_2\}$, $Q_{i_1,i} = I$, $Q_{i_2,i} = I$, and $Q_{i,i} = 2I$, so~\cref{eq:previous_condition} holds with equality. Under this setup,~\cref{eq:previous_claim} becomes
    \begin{equation*}
        \normm{z}_I + \normm{z}_I - \normm{z}_{2I} \leq 0
    \end{equation*}
    which can be rewritten as
    \begin{equation*}
        2 \normm{z} \leq \sqrt{2} \normm{z}
    \end{equation*}
    which is clearly not true. Similarly, other papers, like~\cite{wang2023dmpc} for example, use~\cref{eq:previous_condition} to claim that 
    \begin{equation}
        \label{eq:previous_claim_2}
        \littlesumx_{j \in \calS_i} \normm{Q_{j,i} z} - \normm{Q_{i,i} z} \leq 0,
    \end{equation}
    which is also not true in general. For example, using the same setup as before, if we assume the norms are Euclidean norms and take $z = (1, 1)$, $Q_{i_1, i} = \diag{(3, 2)}$, $Q_{i_2, i} = \diag{(2, 1)}$, and $Q_{i,i} = \diag{(5, 3)}$, we see that~\cref{eq:previous_condition} is satisfied with equality. Then~\cref{eq:previous_claim_2} becomes
    \begin{align*}
        \normm{Q_{i_1, i} z}_2 & + \normm{Q_{i_2, i} z}_2 - \normm{Q_{i, i} z}_2 \\ 
        = & \normm{\bmat{3 \\ 2}}_2 + \normm{\bmat{2 \\ 1}}_2 - \normm{\bmat{5 \\ 3}}_2 \\
        = & \sqrt{13} + \sqrt{5} - \sqrt{34} \approx .01 \geq 0
    \end{align*}
    so~\cref{eq:previous_claim_2} is clearly not true. However, by switching to using $\norm{Qz}$ with arbitrary norms and enforcing $Q = qI$, we achieve the cost function formulation provided in this paper, which results in the correct sufficient condition.
\end{remark}

\begin{remark}
    It is not difficult to extend the condition presented in~\cref{thm:stability_condition} to the case of diagonal matrices when the norms used are the $\ell_1$-norm, \ie, the norms in the cost take the form $\norm{Qy}_1$ where $Q$ is diagonal. In this case, the penalty placed on deviations with respect to position error and velocity error can be weighted differently, which may be desirable in practice. To improve clarity, we elected to stick with a positive scalar weight and the use of arbitrary norms (the equivalent of $\norm{qIx}$).
\end{remark}


\section{EXPERIMENTS}

\subsection{Simulation Experiments}

We simulate a platoon with $N = 50$ follower vehicles and a virtual leader that starts at $p_0(0) = 0\,$m. The leader follows a velocity profile given by
\begin{equation*}
    v_0(t) = \begin{cases}
        20 + t \, \text{m/s} & 0 \leq t \leq 2 \, \text{s} \\
        22 \, \text{m/s} & t > 2 \, \text{s}.
    \end{cases}
\end{equation*}
All experiments use a discrete timestep of $\dt = 0.1\,$s and a prediction horizon of $6$ seconds or $H = 60$. We evaluate only the PF and BD communication topologies for the simulation experiments. Note that under these communication topologies with $\dt = 0.1\,$s and $N = 50$ vehicles, based on~\cref{thm:terminal_state_convergence} it takes 5 seconds or 50 timesteps for the last vehicle's predicted terminal output, $y_{N,t}^p(H)$, to settle on the true desired output relative to the leader.

Experiments using the CDH policy use a desired distance of $5\,$m for $i = 2, \ldots, N$ and $0\,$m for $i = 1$. Experiments using the CTH policy use a time headway of $0.2\,$s and a safety distance of $1\,$m for $i = 2, \ldots, N$ and $0\,$s and $0\,$m for $i = 1$. For all simulation experiments, the vehicles start at the same velocity ($20\,$m/s) and at their desired position with respect to the leader. The dynamics parameter, $\tau_i$, was selected randomly in the range $[0.25, 0.9]$ for each $i$ and each vehicle was constrained to select control inputs in $[-3, 3]\,$m/s$^2$. All vehicles used $q_{i,i} = 1$ and $r_i = 1$. In the PF case, $q_{i,i-1} = 1$ for all $i = 1,\ldots, N$ and in the BD case, $q_{i,i-1} = q_{i,i+1} = 0.5$ for $i = 1, \ldots, N-1$ and $q_{N, N-1} = 1$.

\begin{figure}
    \centering
    \includegraphics{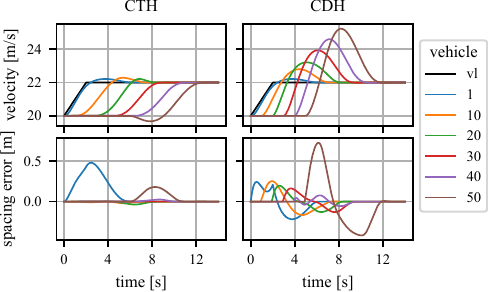}
    \caption{Platoon velocity trajectory and spacing errors for the virtual leader and 6 of the followers from the $N = 50$ vehicles when using a CTH and CDH spacing policy.}
    \label{fig:platoon_velocity_spacing_errs}
    \vspace{-4.0mm}
\end{figure}

We first investigate the use of CTH and CDH spacing policies under a PF topology when using the $\ell_1$-norm in the cost function. The top row of \cref{fig:platoon_velocity_spacing_errs} shows the velocity trajectories of the virtual leader and a subset of the follower vehicles under this scenario. The bottom row shows the spacing errors from each follower vehicle to its predecessor. In the CDH case, the vehicles further from the leader have a tendency to overshoot the desired velocity. In the CTH case, however, the vehicles further from the leader slow down before speeding up and do not overshoot much. We note that this trend was similar for the BD case and $\ell_2$-norm. 

\Cref{fig:platoon_velocity_spacing_errs} also shows that the relative spacing error tends to increase for vehicles further down the platoon (except for vehicle $i=1$ whose error is calculated with respect to the virtual leader). We see that the performance is better in the CTH case compared to the CDH case. We note that if the change to the virtual leader's velocity occurs later (\eg, the virtual leader begins accelerating after $t=5\,$s), the CDH policy performs better since most vehicles can learn the platoon's goal velocity before the leader begins accelerating. This is not practical, however, as it requires the virtual leader (which is likely set by vehicle 1) to know how many vehicles are in the platoon a priori. Moreover, in less predictable scenarios like emergency braking, the lead vehicle may need to decelerate rapidly and will not be able to wait a few seconds to provide follower vehicles enough time to respond to the sudden change in platoon velocity.

\begin{figure}
    \centering
    \includegraphics{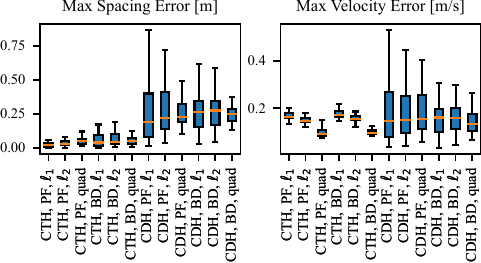}
    \caption{Box plots depicting the spread of the worst errors for each vehicle in the platoon (not including the first vehicle). The notation ``quad'' indicates the weighted squared two norm formulation in the cost.}
    \label{fig:worst-case_box_plot}
    \vspace{-4.0mm}
\end{figure}

To compare more experimental conditions, we investigate the worst performance with respect to spacing and velocity error. For this experiment, we include the use of a quadratic cost (\eg, norms in the cost~\cref{eq:dmpc_cost} are replaced by $||\cdot||_Q^2$ where $Q = qI$). Though this formulation has not been proven stable for general communication topologies (to our knowledge), it has been shown to work well for large platoons~\cite{shaham2024design}. To analyze the performance of CTH vs. CDH, PF vs. BD, and $\ell_1$ vs $\ell_2$ vs quadratic cost (10 possible combinations), we analyzed the maximum errors of each vehicle with respect to predecessor spacing and velocity. The results are shown as a box plot in~\cref{fig:worst-case_box_plot}, where the plot depicts the spread of maximum spacing and velocity errors witnessed by all of the follower vehicles in the platoon. We see that the CTH policy has less variance and performs better with respect to spacing error. Again, we reiterate that if we allowed the virtual leader to start accelerating later, CDH can actually perform as well as or better than CTH. However, this scenario is not always practical.
\begin{figure}[ht]
    \centering
    \includegraphics[width=0.23\textwidth]{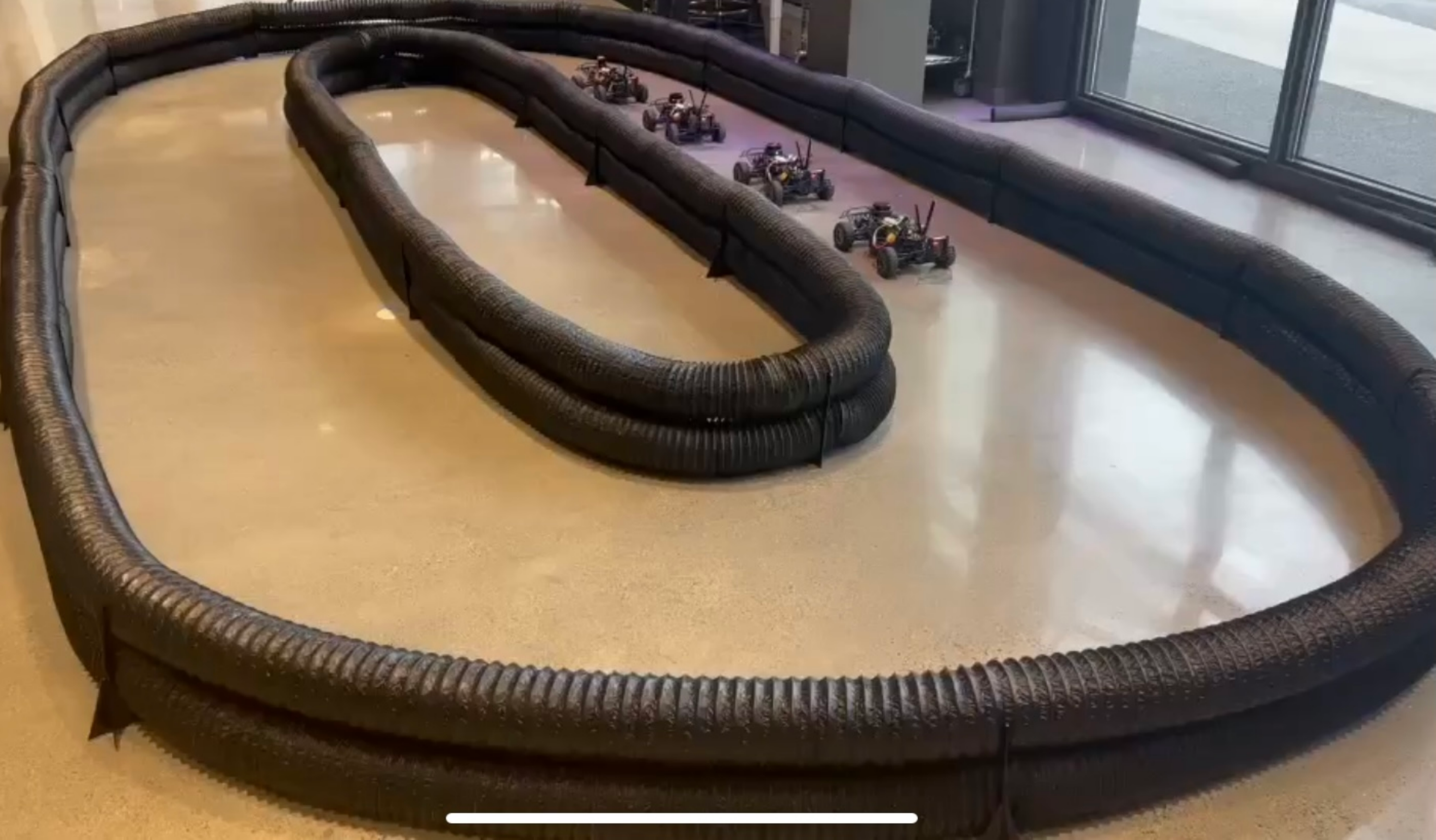}
    \caption{Four F1Tenth vehicles in a $4\,$m $\times$ $8\,$m oval racetrack.}
    \label{fig:convoy_exp}
\end{figure}


\subsection{Hardware Experiments}

To validate the DMPC controllers on hardware, we conduct experiments using four F1Tenth vehicles~\cite{okelly2020f1tenth} in an oval racetrack, seen in \cref{fig:convoy_exp}. The vehicles were modified with ultra-wideband radio frequency (UWB RF) devices to obtain inter-vehicle distance measurements. See~\cite{shaham2024design} for more details on the testing environment and the modifications made to the original F1Tenth platform. Similar to the simulation experiments shown in \cref{fig:platoon_velocity_spacing_errs}, we evaluated the CTH and CDH spacing policies when using a PF communication topology and the $l_1$-norm in the cost function. All algorithm parameters from the previous section remain the same. For the CTH case, we used a time headway of $0.2\,$s and a safety distance of $0.75\,$m. For the CDH case, we used a desired distance of $1\,$m.

\begin{figure}
    \centering
    \includegraphics{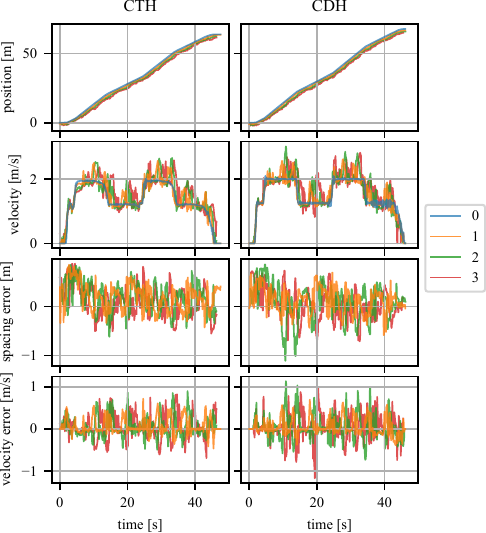}
    \vspace{-2.0mm}
    \caption{Results from two hardware experiments with four F1Tenth vehicles using CTH (left) and CDH (right) under a PF communication topology with the $\ell_1$ norm in the cost function.}
    \label{fig:hardware_results}
    \vspace{-5.0mm}
\end{figure}

Results from the hardware experiment are shown in \cref{fig:hardware_results}. The leader, vehicle 0, follows a velocity profile that alternates between $1.25$ and $2\,$m/s twice before stopping. The top two rows show the platoon position and velocity trajectories for each of the four vehicles. The bottom two rows show the spacing and velocity errors of each follower relative to its predecessor. The performance for CTH and CDH appears similar. The root-mean-square error (RMSE) for the spacing and velocity errors relative to the predecessor are shown in \cref{tab:rmse}. For this experiment with only four vehicles, CDH generally performed better, though the reason for this is unclear. However, we anticipate that in practice, with reasonable choices for CDH and CTH spacing parameters, CTH will offer greater safety because it allows for larger inter-vehicle spacing at higher speeds.

\begin{table}
\vspace{2.0mm}
\caption{Spacing and Velocity RMSE}
\begin{center}
\begin{tabular}{||c||c|c||c|c||}
    \hline
    & \multicolumn{2}{c||}{Spacing RMSE} & \multicolumn{2}{c||}{Velocity RMSE} \\
    \hline
    Vehicle index & CTH & CDH & CTH & CDH \\
    \hline
    1 & 0.33 & 0.2 & 0.29 & 0.2 \\
    \hline
    2 & 0.35 & 0.25 & 0.38 & 0.30 \\
    \hline
    3 & 0.33 & 0.28 & 0.32 & 0.31 \\
    \hline
\end{tabular}
\end{center}
\vspace{-4.0mm}
\label{tab:rmse}
\end{table}


\section{CONCLUSIONS}

This paper introduces a DMPC algorithm for autonomous vehicle platoons with arbitrary communication topologies. The algorithm allows for the use of CTH or CDH spacing policies and can be extended to spacing policies that are affine in the state. Using a Lyapunov argument, the controller is provably asymptotically stable if certain conditions on the weights in the DMPC cost function are satisfied. 

There are a few avenues for potential future research. First, similar to the linear feedback case, it may be possible to develop conditions such that using a DMPC controller with a CTH spacing policy ensures string stability for the platoon. We do not, however, expect this to be the case for a CDH spacing policy as ensuring errors do not propagate without allowing desired distances to change does not seem feasible. Additionally, this algorithm does not assume uncertainty in the state and output variables, and would benefit from explicit handling of uncertainty. Finally, in the real world, communication is lossy, and packets may be dropped. Guaranteeing performance under dropped communication is necessary for real-world implementation.

\addtolength{\textheight}{-12cm}   


\section*{ACKNOWLEDGMENT}

Research was sponsored by the DEVCOM Analysis Center and was accomplished under Cooperative Agreement Number W911NF-22-2-001. The views and conclusions contained in this document are those of the authors and should not be interpreted as representing the official policies, either expressed or implied, of the Army Research Office or the U.S. Government. The U.S. Government is authorized to reproduce and distribute reprints for Government purposes notwithstanding any copyright notation herein.


\bibliographystyle{IEEEtran}
\bibliography{IEEEabrv,refs}

\end{document}